\documentclass[a4paper,USenglish]{lipics-v2018}
\usepackage{microtype}
\usepackage{xspace}
\usepackage{cite}
\usepackage{adjustbox}
\usepackage{tikz}
\hypersetup{hidelinks,
  pdfauthor={S. Beyer, M. Chimani, and J. Spoerhase},
  pdftitle={A Simple Primal-Dual Approximation Algorithm for 2-Edge-Connected Spanning Subgraphs}}
\usepackage[nosemicolon,linesnumbered,vlined,ruled]{algorithm2e}
\SetKwInOut{Input}{Input}
\SetKwInOut{Output}{Output}
\SetKwInOut{Data}{Data}
\hideLIPIcs
\nolinenumbers

\usetikzlibrary{shapes.geometric}

\newcommand{\NNFrac}{\ensuremath{\mathbb{R}_{\geq 0}}\xspace}
\newcommand\pathset[3]{\ensuremath{#1[#2\text{\adjustbox{scale={0.65}{1}}{\textasciitilde\!\textasciitilde\,}}#3]}}
\newcommand\buildpath[1]{\ensuremath{\langle #1 \rangle}}
\newcommand\cyclewitness[2]{\ensuremath{(#1, #2)}}
\newcommand\badcycle[2]{\ensuremath{\cyclewitness{#1}{\buildpath{#2}}}}
\newcommand\pathwitness[2]{\ensuremath{(#1, #2)}}
\newcommand\badpath[2]{\ensuremath{\pathwitness{#1}{\buildpath{#2}}}}

\newcommand{\sitE}{\textsf{E}\xspace}%
\newcommand{\sitU}{\textsf{U}\xspace}%
\newcommand{\sitD}{\textsf{D}\xspace}%

\newcommand{\PiX}[1]{\ensuremath{\Pi^{#1}}}
\newcommand{\PiD}{\PiX{\sitD}}
\newcommand{\PiU}{\PiX{\sitU}}

\newcommand{\aX}[1]{\ensuremath{a^{#1}}}
\newcommand{\aD}{\aX{\sitD}}

\newcommand{\mapD}{\ensuremath{\mu}}
\newcommand{\mapU}{\ensuremath{\eta}}

\newcommand{\startsD}{\ensuremath{\mathbb{S}}}
\newcommand{\sidx}[1]{\ensuremath{{s(#1)}}}

\def\bigO{\ensuremath{\mathcal{O}}}
\def\low{\operatorname{low}}

\makeatletter
\def\th@claim{%
  \thm@headfont{%
    \textcolor{darkgray}{$\vartriangleright$}\nobreakspace\sffamily}%
  \normalfont %
  \thm@preskip\topsep \divide\thm@preskip\tw@
  \thm@postskip\thm@preskip
}
\newenvironment{claimproof}[1][\normalfont\sffamily{}\proofname]{\par
  
  \pushQED{\qed}%
  \normalfont \topsep0\p@\@plus6\p@\relax
  \trivlist
  \item[\hskip\labelsep
        \color{darkgray}\sffamily\bfseries
    #1\@addpunct{.}]\ignorespaces
}{%
  \popQED\endtrivlist%
}
\makeatother
\theoremstyle{claim}
\newtheorem{claim}{Claim}
\title{A Simple Primal-Dual Approximation Algorithm for 2-Edge-Connected Spanning Subgraphs}
\titlerunning{A Simple Primal-Dual Approximation Algorithm for 2ECSS}

\author{Stephan Beyer}{Theoretical Computer Science, Osnabrück University, Germany}{stephan.beyer@uni-osnabrueck.de}{0000-0001-5274-0447}{}
\author{Markus Chimani}{Theoretical Computer Science, Osnabrück University, Germany}{markus.chimani@uni-osnabrueck.de}{0000-0002-4681-5550}{}
\author{Joachim Spoerhase}{Department of Computer Science, Aalto University, Finland}{joachim.spoerhase@aalto.fi}{0000-0002-2601-6452}{}
\authorrunning{S. Beyer, M. Chimani, and J. Spoerhase}

\Copyright{Stephan Beyer, Markus Chimani, and Joachim Spoerhase}

\funding{%
 The first two authors are supported
  by the German Research Foundation~(DFG), grant CH~897/3-1.
 The third author is supported
  by the European Research Council (ERC) under the European Union’s Horizon 2020 research and innovation programme, grant~759557,
  and by the Academy of Finland, grant~310415.}

\subjclass{Mathematics of computing $\rightarrow$ Approximation algorithms, %
  Mathematics of computing $\rightarrow$ Paths and connectivity problems}
\EventLongTitle{}
\EventShortTitle{}
\EventAcronym{}
\EventYear{2018}
\EventDate{}
\EventLocation{}
\begin{document}

\maketitle

\begin{abstract}
  We propose a simple and natural approximation algorithm
   for the problem of finding a 2-edge-connected spanning subgraph
   of minimum total edge cost in a graph.
  The algorithm maintains a spanning forest starting with an empty edge set.
  In each iteration, a new edge incident to a leaf is selected in a natural greedy manner and added to the forest.
  If this produces a cycle, this cycle is contracted.
  This growing phase ends when the graph has been
  contracted into a single node and a subsequent cleanup step removes redundant edges in reverse order.

  We analyze the algorithm using the primal-dual method
   showing that its solution value is at most 3 times the optimum.
  Although this only matches the ratio of existing primal-dual algorithms,
   we require only a single growing phase,
   thereby addressing a question by Williamson.
  Also, we consider our algorithm to be
   not only conceptually simpler than the known approximation algorithms
   but also easier to implement in its entirety.
  For $n$ and $m$ being the number of nodes and edges, respectively,
   it runs in $\bigO(\min\{nm, m + n^2 \log n\})$ time
   and $\bigO(m)$ space
   without data structures more sophisticated than binary heaps and graphs,
   and without graph algorithms beyond depth-first search.
\end{abstract}

\section{Introduction}

An undirected multigraph $G = (V,E)$
 with $n := |V|, m := |E|, m \geq n$,
 is \emph{2-edge-connected}
 if for every edge $e \in E$ the graph
$G - e := (V, E \setminus \{e\})$ is connected.
The \emph{minimum 2-edge-connected spanning subgraph} problem~(2ECSS) is defined
as follows: Given a 2-edge-connected undirected multigraph $G = (V,E)$
with edge costs $c\colon E \to \NNFrac$, find an edge subset $E'\subseteq E$ of minimum cost
$c(E') := \sum_{e \in E'} c(e)$ such that
 $G' = (V, E')$ is 2-edge-connected.
Any edge of $G$ may only be used once in $G'$.
2ECSS is a fundamental NP-hard network design problem
 that arises naturally in the planning of infrastructure
 where one wants to guarantee a basic fault tolerance.

\subparagraph*{Related work.}

Some algorithms mentioned below
 work not only for 2ECSS but
 for more general problems,
 like $k$ECSS with $k\geq 2$.
Since we are interested in the former,
 we describe their results,
 in particular the achieved approximation ratios,
 in the context of 2ECSS.
We restrict our attention to algorithms
 able to work on general edge costs
 (in contrast to, e.g., metric or Euclidean edge costs).

The first algorithms \cite{FJ81,KT93}
 yield $3$-approximations
 by using a minimum spanning tree in $G$
 and augmenting it to become 2-edge-connected.
The factor $3$ is based on
 a $2$-approximation for the latter problem (often called \emph{weighted tree augmentation}).
The algorithm of~\cite{FJ81} runs in $\bigO(n^2)$ time
 and that of \cite{KT93} in $\bigO(m + n \log{n})$.

In~\cite{KV94},
 a $2$-approximation algorithm
 is obtained
 by reducing the problem
 to a weighted matroid intersection problem
 that can be solved
 in time $\bigO(n (m + n \log n) \log n)$~\cite{G95}.
The $2$-approximation algorithm in~\cite{J01}
 is based on iterative rounding of solutions to a linear programming formulation
 that we will see in a later section.
On the negative side,
 no algorithm with factor less than $2$ is known,
 and,
 unless P\,$=$\,NP,
 there cannot be a polynomial-time approximation with ratio better than
 roughly $1+\frac1{300}$~\cite{P10}.

Besides the algorithms mentioned above,
 there is a separate history
 of applying the primal-dual method.
The basic idea of a primal-dual algorithm
 is that a feasible solution to the dual of the aforementioned linear program is computed
 and this process is exploited to compute an approximate primal solution.
There are several primal-dual $3$-approximation algorithms~\cite{KR93,SVY92,WGMV95,GGW98,GGPSTW94}
  with the best running time being $\bigO(n^2 + n \sqrt{m \log\log n})$.
All algorithms
 grow a solution in two phases:
 they first obtain a spanning tree,
 and then augment that tree to be 2-edge-connected.
Then, unnecessary edges are deleted in a cleanup phase to obtain minimality.
Most algorithms are algorithmically complex
 and, for example,
 require solving multiple network flow problems.

\subparagraph*{Contribution.}

We present a \emph{simple} 3-approximation algorithm,
 analyzed using primal-dual techniques,
 that finds a minimal 2-edge-connected spanning subgraph on general edge costs.

In comparison to the other primal-dual algorithms,
 we grow the solution in a single phase, i.e,
 we omit obtaining an intermediate spanning tree.
Thus we make progress on the question by Williamson~\cite{W93}
 if it is possible ``to design a single phase algorithm for some class of edge-covering problems''.
The (to our best knowledge) new conceptual idea is to
 modify the classical synchronized primal-dual scheme by growing the solution
 only at leaves of the current solution.
We contract arising 2-edge-connected components on the fly.
Although we do not beat the so-far best primal-dual approximation ratio $3$,
 our algorithmic framework may offer new insight for further improvements.

Moreover,
 our algorithm is conceptually much simpler
 than the aforementioned approximation algorithms.
In contrast to the 2-approximation algorithms
 based on
 weighted matroid intersection
 or linear programming,
 our algorithm requires only trivial data structures (arrays, lists, graphs, and optionally binary heaps)
 and no graph algorithms beyond depth-first search.
It is simple to implement in its entirety
 to run in $\bigO(\min\{nm, m + n^2 \log n\})$~time
 while occupying only $\bigO(m)$~space.

\subparagraph*{Preliminaries.}

We always consider an undirected multigraph $G = (V, E)$
 with node set $V$ and edge set $E$.
As we allow parallel edges, we identify edges by their names, not by their incident nodes.
For each $e \in E$,
 let $V(e) := \{v, w\} \subseteq V$
 be the two nodes incident to $e$.
We may describe subgraphs of $G$
 simply by their (inducing) edge subset $H \subseteq E$.
By $V(H) := \bigcup_{e \in H}{V(e)}$
 we denote the set of nodes spanned by the edges of $H$.
For each $v \in V(H)$,
 let $\delta_H(v) := \{e \in H \mid v \in V(e) \}$ be the edges incident to $v$.
For any $S \subsetneq V(H)$,
 let $\delta_H(S) := \{e \in H \mid V(e) = \{u, v\}, u \in S, v \notin S\}$.
The \emph{degree} $\deg_H(v) := |\delta_H(v)|$ of $v \in V$ in $H$
 is the number of incident edges of $v$ in $H$.

A \emph{path} $P$
 of length $k \geq 0$
 is a subgraph
 with $P = \{e_1, \ldots, e_k\}$
 such that
  there is an orientation of its edges where
  the head of $e_i$ coincides with the tail of $e_{i+1}$ for $i < k$.
In such an orientation,
 let $u$ be the tail of $e_1$
 and $v$ the head of $e_k$.
We call $u$ and $v$ the \emph{endpoints} of $P$,
 and $P$ a
 \emph{$u$-$v$-path}
 (or, equivalently, a \emph{path between $u$ and $v$}).
Observe that
 our definition of paths
 allows nodes but not edges to repeat
 (due to set notation).
A path $P$ is \emph{simple}
 if and only if $\deg_P(v) \leq 2$ for all $v \in P$.
For a simple $u$-$v$-path $P$,
 we call $V(P) \setminus \{u, v\}$
 the \emph{inner nodes} of $P$.
We call a path \emph{closed}
 if both endpoints coincide
 (i.e., if it is a $u$-$v$-path with $u = v$),
 and \emph{open} otherwise.
A \emph{cycle} is a closed path of length at least $2$.
We say two paths $P_1, P_2$ are \emph{disjoint}
 if and only if $P_1 \cap P_2 = \varnothing$,
 i.e., they do not share a common edge
 (they may share nodes).

Let $G$ be 2-edge-connected.
An edge $e \in E$ is \emph{essential}
 if and only if $E \setminus \{e\}$ is not 2-edge-connected;
 it is \emph{nonessential} otherwise.
An \emph{ear} is a simple path $P$ of length at least $1$
 such that $E \setminus P$ is 2-edge-connected.

For any function $f\colon A \to B$ and any $A' \subseteq A$, we
 denote by $f(A') := \{f(a)\mid a \in A'\}$ the image of $A'$ under
 $f$ (unless otherwise stated).
We also define
 $f^{-1}(b) := \{a \in A \mid f(a) = b\}$.

\subparagraph*{Organization of the paper.}

Although our algorithm (which is described in Section~\ref{section:alg})
 turns out to be surprisingly simple,
 its analysis is more involved.
In Section~\ref{section:analysis},
 we will show its
 time and space complexity
 as well as its approximation ratio
 (under the assumption that a particular \emph{leaf-degree property}, which may be of independent interest, holds in every step of the algorithm).
The technical proof of the leaf-degree property
 is deferred to Section~\ref{section:degree-property}, where it is shown independently (also to simplify the required notation).

\section{The Algorithm}
\label{section:alg}

Our algorithm is outlined in Algorithm~\ref{alg}.
Given a (multi-)graph $G=(V,E)$ with cost function $c\colon E \to \NNFrac$,
 the main \emph{grow phase}
 selects edges $T\subseteq E$
 such that $T$ is spanning and 2-edge-connected,
 but not necessarily minimal.
The central idea of the grow phase---in contrast to several other primal-dual approaches---is to only grow the solution with edges that are currently attached to leaves.\footnote{%
This is a key difference to the \emph{second} phase suggested in \cite{KR93}, which on first sight looks somewhat similar (but leads to very different proof strategies).
In particular,
 we can directly attack the 2-edge-connected subgraph
 in a single phase,
 instead of a multi-phase growing procedure where each phase has to consider distinct objectives and rules.}
Afterwards, a trivial \emph{cleanup phase}
 removes nonessential edges from $T$,
 checking them in reverse order,
 to obtain the final solution.

\begin{algorithm}[tbp]
 graph $G' = (V', E')$ with edge costs $c':=c$ as a copy of $G=(V, E)$\;
 solution $T := \varnothing$\;
 forest $F := (V',\varnothing)$\;
 \BlankLine
 \While(\tcp*[f]{grow phase}){$F$ is not a single node}{
   Simultaneously for each leaf in $F$, decrease the cost $c'$ of its incident edges in $G'$
    until an edge, say $\tilde e$, gets cost $0$;
    by this, an edge cost $c'$ is reduced with double speed if it is incident to two leaves\;
    \label{line:delta}
   Add $\tilde e$ to $F$ and to $T$\;
    \label{line:add}
   \lIf{$\tilde e$ closes a cycle $Q$ in $F$}{contract $Q$ in $F$ and in $G'$}
    \label{line:cycle}
 }
 \BlankLine
 \ForAll(\tcp*[f]{cleanup phase}){$e \in T$ in reverse order}{%
  \lIf{$T - e$ is 2-edge-connected}{%
   remove $e$ from $T$
    \label{line:cleanup}
  }
 }
 \caption{Approximation algorithm for 2ECSS}%
 \label{alg}%
\end{algorithm}

The rest of this paper focuses on proving
 our main Theorem~\ref{thm:main} below.

\begin{theorem}\label{thm:main}
 There is an algorithm for 2ECSS
  that runs in $\bigO(\min\{n m, m + n^2 \log n\})$ time
  and $\bigO(m)$ space.
 It obtains solutions within three times the optimum.
\end{theorem}

\section{Analysis of Algorithm~\ref{alg} (Proof of Theorem~\ref{thm:main})}
\label{section:analysis}

We call the iterations within the phases of Algorithm~\ref{alg}
 \emph{grow steps} and \emph{cleanup steps}, respectively.
In a grow step, a cycle may be contracted and some edges become incident to the contracted node.
As we identify edges by their names, the names of these edges are retained although their incident nodes change.

Let $E'$ and $E_F$ be the edge set of $G'$ and $F$, respectively.
During the algorithm we have the following invariants:
 both $G'$ and $F$ use the common node set $V'$ that describes a partition of~$V$;
 we consider $T$ to form a subgraph of $G$;
 each edge in $E_F$ represents an edge of $T$
 that is not part of a cycle in $G$;
 we have~$E_F \subseteq E' \subseteq E$.

Initially, each node of $V$ forms an individual partition set,
 i.e., $|V'| = |V|$, and $E_F = \varnothing$.
We merge partition sets (nodes of $V'$, cf.\ line~\ref{line:cycle})
 when we contract a cycle, i.e.,
 when the corresponding nodes in $V$ induce a 2-edge-connected subgraph in $T$.
Arising self-loops are removed both from $G'$ and $F$.
The grow phase terminates once $|V'| = 1$,
 i.e., all nodes of $V$ are in a common 2-edge-connected component.

Observe that for an edge $e \in E'$, we naturally define
 $V(e) \subseteq V$ as its incident nodes in original $G$, and
 $V'(e) \subseteq V'$ as its incident nodes in $G'$ and $F$.

Let $L := \{v \in V' \mid \deg_F(v) \leq 1\}$
 be the set of leaves (including isolated nodes) in $F$.
For any edge $e \in E'$,
 let $\ell_e := |V'(e) \cap L| \in \{0, 1, 2\}$
 be the number of incident nodes of $e$ that are leaves in~$F$.
An edge $e \in E'$ is \emph{eligible}
 if
 $e \notin T$ and
 $\ell_e \geq 1$.
Let $\Delta(e) := \frac{c'(e)}{\ell_e}$ for eligible edges $e \in E'$.
Now line~\ref{line:delta}
 can be described as
 first finding
 the minimum (w.r.t.\ $\Delta$) eligible edge $\tilde e \in E$,
 and then, for each eligible edge $e$,
 decreasing $c'(e)$ by $\ell_e \Delta(\tilde e)$.
For convenience,
 we denote $\Delta(\tilde e)$ by $\tilde\Delta$.

\subsection{Time and Space Complexity}

Here we state the main time and space complexity results
 for Algorithm~\ref{alg}.
\begin{lemma}\label{lemma:time}
 Algorithm~\ref{alg} can be implemented to run
  in $\bigO(n m)$ time
  and $\bigO(m)$ space
  using only arrays and lists.
\end{lemma}
\begin{proof}
 Let us first describe the used data structures.
 The graphs $G'$ and $F$ are stored naturally using the adjacency list representation.
 We store $T$ as a list (or array)
 and manage the costs~$c'$ in an array.
The space consumption of $\bigO(m)$ follows directly,
 and the initialization of all data structures takes $\bigO(m)$ time.

 We now show that there are $\bigO(n)$ grow and cleanup steps.
 Consider $T$ directly after the grow phase.
 Let $T_0 \subsetneq T$
  be the edges
  that led to a contraction,
  thus $|T_0| < n$.
 The edges $T \setminus T_0$
  form a tree
  since any cycle would lead to a contraction.

 Line~\ref{line:delta}
  takes $\bigO(m)$
  time by
  iterating over all $e \in E'$ twice,
  first to find $\tilde e$,
  and
  second to reduce the costs $c'$.
 Line~\ref{line:add}
  adds $\tilde e$ to $T$
  and to $F$ in constant time.
 For line~\ref{line:cycle},
  we can find the respective cycle in $F$
  (or determine that it does not exist)
  in $\bigO(n)$ time
  using depth-first search.
 The contractions of the (same) cycle in $G'$ and $F$ are performed
  in $\bigO(n)$ time.

 For the running time of a cleanup step (line~\ref{line:cleanup}),
  consider $T$ at the current iteration.
 Note that $T$ induces a 2-edge-connected graph.
 We can check in $\bigO(|T|)$ time
  if an edge $e \in T$
  with $V(e) = \{s, t\}$
  is essential
  using a simplified version of the classical 2-connectivity test:
 Perform a depth-first search
  in $T - e$
  starting at $s$
  to compute the DFS indices for each node.
 Using a single bottom-up traversal,
  compute $\low(u)$ for each $u \in V$
  where $\low(u)$ is the node
  with the smallest DFS index reachable
  from $u$
  by using only
  edges in the DFS tree to higher DFS indices,
  and
  at most one edge not in the DFS tree.
 Now initialize $v := t$
  and iteratively go to $v := \low(v)$ until
   $v = \low(v)$.
 Clearly,
  there is a cycle in $T - e$ containing $s$ and $t$
  if and only if $v = s$;
  otherwise $e$ is essential.
\end{proof}

Unfortunately,
 $m$ can be unbounded if $G$ has parallel edges.
We can improve the running time
 for graphs with
 $m \in \omega(n \log{n})$.
For this, we first shrink $m$ to $\bigO(n^2)$
 by removing uninteresting (parallel) edges
 in the beginning as well as after every contraction.
Then,
 we use one global value $\Gamma$
 that simulates the shrinking of all costs~$c'$
 in constant time,
 and we use $\bigO(\sqrt{m})$ binary heaps
 of size $\bigO(\sqrt{m})$
 to manage the eligible edges and their costs $c'$;
 $\bigO(n)$ many updates in these binary heaps
 lead to a dominating running time of $\bigO(n \log n)$ per grow step.

\begin{lemma}\label{lemma:time-dense}
 Algorithm~\ref{alg} can be implemented to run
  in $\bigO(m + n^2 \log{n})$ time
  and $\bigO(m)$ space
  using only arrays, lists, and binary heaps.
\end{lemma}
\begin{proof}
 Consider a set of $p > 2$ pairwise parallel edges.
 We call each of the $p - 2$ highest-cost edges
  \emph{futile}.
 By removing all futile edges
  we decrease the maximum edge-multiplicity in $G$ to $2$.
 This reduction can be performed in $\bigO(m + n^2)$ time (e.g., by using $n$ buckets for the neighbors of each $v \in V$).
 This guarantees
  that each node has degree $\bigO(n)$,
  hence from now on $m \in \bigO(n^2)$.
 Whenever we contract a cycle $Q$ of $n_Q$ nodes
  into a single node $q$ in line~\ref{line:cycle},
  we re-establish $\deg_{G'}(q) \in \bigO(n)$
  by removing arising self-loops and futile edges.
 Spotting futile edges requires $\bigO(n)$ time per contraction.
 Although there are
  $\bigO(n_Q^2)$ edges within $V(Q)$ (which become self-loops at $q$)
  and
  $\bigO(n_Q n)$ edges from $V(Q)$ to $V_G \setminus V(Q)$,
  the removal of the self-loops and futile edges
  takes $O(m)$ time in total (for the whole grow phase),
  since any edge is removed at most once.

 Instead of storing $c'$ directly,
  store a single global value $\Gamma$, initially being zero,
  and for each edge $e \in E'$, store a value $\bar \Delta(e)$
  that keeps the invariant
  $\bar\Delta(e) = \Delta(e) + \Gamma$
   if $e$ is eligible
  and
  $\bar \Delta(e) = c'(e)$
   if not.
 Furthermore,
  we partition the eligible edges of $E'$ arbitrarily
  into $\Theta(\sqrt{m})$ many (binary) heaps
  of size $\Theta(\sqrt{m})$,
  with $\bar\Delta$ as the priorities.
 Initially, all edges are eligible.
 The initialization takes $\bigO(m)$~time.

 By looking at the minimum edge of each heap,
  we can find and extract $\tilde e$ (see line~\ref{line:delta})
  in $\bigO(\sqrt{m})$ time.
 To decrease $c'(e)$ by $\ell_e \tilde\Delta$ for all eligible $e \in E'$,
  we have to decrease $\Delta(e)$ by $\tilde\Delta$
  which is performed by increasing $\Gamma$ by $\tilde\Delta$
  in constant time.
 Observe that this preserves the invariant:
  If $e$ is eligible,
   decreasing $c'(e)$ by $\ell_e \tilde\Delta$
   updates $\Delta(e)$ to $\frac{c'(e) - \ell_e \tilde\Delta}{\ell_e} = \frac{c'(e)}{\ell_e} - \tilde\Delta$,
   which is exactly what increasing $\Gamma$ by $\tilde\Delta$ does;
  otherwise,
   neither $\Delta(e)$ nor $\Gamma$ is changed.

 Adding edge $e$ with $V'(e) = \{s, t\}$ to $F$ (line~\ref{line:add})
  or contracting a cycle $Q$ into $q$ (line~\ref{line:cycle})
  may change the $\ell_e$ values of all $\bigO(n)$ edges in $G$
  incident to $s$ and $t$ or to $q$,
  respectively.
 If $\ell_e$ decreases to $0$,
  we remove $e$ from its heap;
  if it increases from $0$,
  we re-add it to its original heap.
 Otherwise,
 if $\ell_e$ decreases from $2$,
  we want $\Delta(e)$ to double;
 if it increases to $2$,
  $\Delta(e)$ should be halved.
 The necessary changes to $\bar \Delta(e)$ follow straightforwardly.
 Each necessary operation on a heap
  requires $\bigO(\log{n})$ time.
 We can hence perform each grow step in $\bigO(n \log n)$~time.
\end{proof}

\subsection{Analysis of the Approximation Ratio}

Let
 $\mathcal{S} := 2^V \setminus \{\varnothing, V\}$
 and
 $\mathcal{S}_{e} := \{S \in \mathcal{S} \mid e \in \delta_G(S)\}$
 for any $e \in E$.
We analyze the approximation ratio using the primal-dual method.
Hence consider the basic integer program for~2ECSS:
\begin{align}
  \textsl{minimize}\quad & \sum_{e \in E}{c(e) x_e} \\
  \sum_{e \in \delta_G(S)}{x_e} & \geq 2 & & \forall S \in \mathcal{S}
  \label{eq:lp:cut}
 \\
  x_{e} & \in \{0, 1\} & & \forall e \in E
  \label{eq:lp:integrality}
  \text{.}
\intertext{%
For its linear relaxation,
 \eqref{eq:lp:integrality} is substituted by $0 \leq x_e \leq 1$ for every $e \in E$.
The bound $x_e \leq 1$ is important since edge duplications are forbidden.
Its dual program is}
  \textsl{maximize}\quad & 2 \sum_{S \in \mathcal{S}}{y_S} - \sum_{e \in E}{z_e} \\
  \sum_{S \in \mathcal{S}_e}{y_S} - z_e & \leq c(e)
   & & \forall e \in E
   \label{eq:dual:cost}%
 \\
  y_S & \geq 0
   & & \forall S \in \mathcal{S}
 \\
  z_e & \geq 0
   & & \forall e \in E
  \text{.}
 \end{align}
We show that Algorithm~\ref{alg}
 implicitly constructs a solution $(\bar y, \bar z)$ to the dual program.
Let $(\bar y^i, \bar z^i)$ denote this dual solution computed after the $i$-th grow step.
Initially, we have the dual solution $(\bar y^0, \bar z^0)=0$.
Following this notion,
 let $F^i = (V^i, E_{F^i})$ be the forest after the $i$-th grow step,
 $L^i := \{v \in F^i \mid \deg_{F^i}(v) \leq 1\}$,
 and $\ell_e^i := |V'(e) \cap L^i|$ for each $e \in E'$.
For any node $v \in V'$,
 let $S(v)$ be the corresponding node subset of~$V$.

\begin{lemma}\label{lemma:algorithm-lp}
 The grow phase constructs a feasible solution to the dual problem implicitly as follows.
 We have, for each $v \in V'$ and $i \geq 1$,
 \begin{align*}
  \bar y_{S(v)}^{i}
   & := \begin{cases}
    \bar y_{S(v)}^{i-1} + \tilde\Delta & \text{if $v \in L^{i-1}$} \\
    \bar y_{S(v)}^{i-1} & \text{otherwise,}
   \end{cases}
 & \quad
  \bar z_e^{i}
  & := \begin{cases}
    \bar z_e^{i-1} + \tilde\Delta & \text{if $v \in L^{i-1}$ and $e \in \delta_{F^{i-1}}(v)$} \\
    \bar z_e^{i-1} & \text{otherwise.}
  \end{cases}
 \end{align*}
\end{lemma}
\begin{proof}
 Let $\bar{c}^i$ be $c'$ after the $i$-th grow step.
 Initially, $(\bar y^0, \bar z^0) = 0$
  matches
  the initialization $\bar{c}^0 := c$.
 Consider the $i$-th grow step.
 We show that
  $(\bar y^i, \bar z^i)$ satisfies
  (i)~$\bar{c}^i(e) = c(e) - \sum_{S \in \mathcal{S}_e}{\bar y_S^i} + \bar z_e^i$,
  i.e., the right-hand side minus the left-hand side of~\eqref{eq:dual:cost},
  and (ii)~$\bar{c}^i(e) \geq 0$,
  i.e., the constructed solution is feasible.

 (i) The claim is trivial for $v \notin L^{i-1}$
  since the corresponding variables do not change.
 Consider $v \in L^{i-1}$
  and any $e \in \delta_G(S(v))$.
 By $\bar y_{S(v)}^i = \bar y_{S(v)}^{i-1} + \tilde\Delta$,
  we have
  $\sum_{S \in \mathcal{S}_e}{\bar y_S^{i-1}} - \bar z_e^{i-1}
    = \sum_{S \in \mathcal{S}_e}{\bar y_S^i} - \bar z_e^{i-1} - \tilde\Delta$
  for the left-hand side of~\eqref{eq:dual:cost}.
 By the definition of $\bar z_e^i$,
  this coincides with
  $\sum_{S \in \mathcal{S}_e}{\bar y_S^i} - \bar z_e^i$
   if $e \in \delta_F(v)$,
  and with
   $\sum_{S \in \mathcal{S}_e}{\bar y_S^i} - \bar z_e^i - \tilde\Delta$
   otherwise.
 This change is reflected exactly
  by $\bar{c}^i(e) := \bar{c}^{i-1}(e) - \ell_e^{i-1} \tilde\Delta$
  (that is, decreasing $c'$ by $\tilde\Delta$ for each leaf incident to $e$ in~$F$)
  if and only if $e$ is eligible.

 (ii) Assume by contradiction
  that
  there is an $e \in E$
  with $\bar{c}^{i-1}(e) \geq 0$ and $\bar{c}^i(e) < 0$.
 Note that $\bar{c}^{i-1}(\tilde e) = \ell_{\tilde e}^{i-1} \tilde\Delta$.
 By $\bar{c}^i(e) := \bar{c}^{i-1}(e) - \ell_e^{i-1} \tilde\Delta < 0$
  we get $\bar{c}^{i-1}(e) < \ell_e^{i-1} \tilde\Delta$,
  which contradicts the choice of $\tilde e$.
\end{proof}

Let $\bar T$ be the solution edges remaining after the cleanup phase.
Let $(V^i, \bar T^i)$
 be the graph on nodes $V^i$
 that consists of all edges in $\bar T$
 without self-loops.
In other words,
 $\bar T^i$
 are the edges corresponding to $\bar T$
 when mapped into the node partition defined by $F^i$.
We partition $L^i$ into
 the set
  $L_0^i := \{v \in L^i \mid \deg_{F^i}(v) = 0\}$
  of isolated nodes in $F^i$,
 the set
  $L_1^i := \{v \in L^i \mid \deg_{F^i}(v) = 1, \delta_{F^i}(v) \subseteq \bar T^i\}$
  of degree-1 nodes in $F^i$
  incident to an edge in the contracted solution $\bar T^i$,
 and
 the set
  $L_2^i := \{v \in L^i \mid \deg_{F^i}(v) = 1, \delta_{F^i}(v) \cap \bar T^i = \varnothing\}$
  of the degree-1 nodes in $F^i$
  incident to an edge in $E^i \setminus \bar T^i$,
  i.e., not being in the contracted solution.

\begin{lemma}
 For each $i$,
  every edge $e \in \bar T^i \setminus E_{F^i}$ is essential in $(V^i, \bar T^i \cup E_{F^i})$.
 \label{lemma:essential}
\end{lemma}
\begin{proof}
 First observe that
  for a cycle~$Q$ in a 2-edge-connected graph~$H$,
  an edge $e \notin Q$
  is essential in $H$ if and only if $e$ is essential in $H$ after contracting $Q$.

 The claim holds trivially for the single-node forest (i.e., after the last grow step).
 Consider any $i$
  where we have that every edge $e \in \bar T^i \setminus E_{F^i}$
  is essential in $(V^i, \bar T^i \cup E_{F^i})$.
 In the $i$-th grow step,
  we insert an edge $e \in \bar T^{i-1} \setminus E_{F^{i-1}}$
  into $F^{i-1}$ or possibly contract an emerging cycle.
 In any way,
  $e \notin \bar T^i \setminus E_{F^i}$.
 By induction,
  all edges $\bar T^{i-1} \setminus (E_{F^{i-1}} \cup \{e\})$ are essential.
 If $e$ was nonessential,
  the cleanup step corresponding to the $i$-th grow step would remove $e$.
\end{proof}

\begin{lemma}[Leaf-Degree Property]\label{lemma:degree-property}%
 We have
 $
  \sum_{v \in L^i}{\deg_{\bar T^i}(v)}
   \leq
   3 (|L^i| + |L^i_0|) + |L^i_1|
 $.
\end{lemma}
Note that Lemma~\ref{lemma:essential} is a prerequisite for Lemma~\ref{lemma:degree-property}.
The proof of Lemma~\ref{lemma:degree-property} is highly non-trivial and may be of independent interest.
We thus defer its presentation,
 together with all the required notations and further definitions,
 to Section~\ref{section:degree-property}.
There, we will restate the lemma (including all prerequisites) in more general terms as Theorem~\ref{thm:degree-property}.

\begin{lemma}\label{lemma:approx-ratio}%
 The solution obtained by Algorithm~\ref{alg} is within three times the optimum.
\end{lemma}
This result is tight as can be seen in Figure~\ref{fig:tight}.
\begin{proof}
 Let $(\bar y, \bar z)$ be the \emph{dual} solution Algorithm~\ref{alg} produces implicitly, as described by Lemma~\ref{lemma:algorithm-lp},
  with dual solution value $B$.
 On the other hand, $\bar T$ is called our \emph{primal} solution.
 Note that for all edges $e \in \bar T$, we have $c'(e) = 0$,
  i.e.,
  their constraints \eqref{eq:dual:cost} are tight.
 Hence we can rewrite our primal solution value
 \[
  c(\bar T) = \sum_{e \in \bar T}{c(e)}
   = \sum_{e \in \bar T}\Big(\sum_{S \in \mathcal{S}_e}{\bar y_S} - \bar z_e\Big)
   = \sum_{S \in \mathcal{S}}{\deg_{\bar T}(S) \bar y_S} - \sum_{e \in \bar T}{\bar z_e}
   \text{.}
 \]
 We prove a $3$-approximation by showing that
  $c(\bar T) \leq 3 B$,
  i.e.,
 \[
  \sum_{S \in \mathcal{S}}{\deg_{\bar T}(S) \bar y_S} - \sum_{e \in \bar T}{\bar z_e}
   \, \leq \,
   3\Big(\sum_{S \in \mathcal{S}}{2 \bar y_S} - \sum_{e \in E}\bar z_e\Big)
  \text{,}
 \]
 or equivalently
 \begin{equation}
  \sum_{S \in \mathcal{S}}{\deg_{\bar T}(S) \bar y_S}
   \, \leq \,
   6 \sum_{S \in \mathcal{S}}{\bar y_S}
   - 2 \sum_{e \in \bar T}\bar z_e
   - 3 \!\! \sum_{e \in E \setminus \bar T} \!\!\! \bar z_e
   \text{.}
  \label{eq:primaldual:essence}%
 \end{equation}
 Observe that~\eqref{eq:primaldual:essence} trivially holds initially
  since all values $(\bar y^0, \bar z^0)$ are zero.
 We show that \eqref{eq:primaldual:essence} holds after each grow step.
 Assume it holds for $(\bar y^i, \bar z^i)$.
 We look at the increase
  of the left-hand side
  and right-hand side
  of \eqref{eq:primaldual:essence}
  when adding an edge to $F^i$.
 By Lemma~\ref{lemma:algorithm-lp},
  we have $\bar y^{i+1}_{S(v)} = \bar y^i_{S(v)} + \tilde\Delta$ for all $v \in L^i$
  and
  $\bar z^{i+1}_e = \bar z^i_e + \tilde\Delta$ for all $e \in \delta_F(L^i_1 \cup L^i_2)$.
 Hence
  it remains to show that
  \[
  \sum_{v \in L^i}{\deg_{\bar T^i}(v) \tilde\Delta}
   \; \leq \;
   6 \sum_{v \in L^i}\tilde\Delta
   - 2 \sum_{v \in L^i_1}\tilde\Delta
   - 3 \sum_{v \in L^i_2}\tilde\Delta
  \]
  holds.
 After dividing by $\tilde\Delta$
  and since $|L^i| = |L^i_0| + |L^i_1| + |L^i_2|$,
  the right-hand side
  simplifies to
  $6 |L^i| - 2 |L^i_1| - 3 |L^i_2|
   =
   6 |L^i_0| + 4 |L^i_1| + 3 |L^i_2|
   =
   3 (|L^i| + |L^i_0|) + |L^i_1|$,
  i.e., we have Lemma~\ref{lemma:degree-property}.
\end{proof}

\section{The Leaf-Degree Property (Proof of Lemma~\ref{lemma:degree-property})}
\label{section:degree-property}

\tikzset{%
  general/.style={%
    draw=black,
    fill=black,
    inner sep=1pt,
    circle,
  },
  gedge/.style={%
    draw=black!70!yellow,
  },
  subgraph/.style={%
    gedge,
    fill=lightgray!40,
    ellipse,
  },
  fedge/.style={%
    very thick,
    draw=black!60!blue,
  },
  not in solution/.style={%
    densely dashed,
  },
  leaf/.style={%
    fedge,
    fill=white,
    inner sep=1mm,
    rectangle,
  },
  filledleaf/.style={%
    fedge,
    fill,
    inner sep=1mm,
    rectangle,
  },
  pics/branch/.style={%
    code={
      \node (#1-n1) [filledleaf] at (1,0) {};
      \node (#1-n2) [filledleaf] at (2,0) {};
      \draw [gedge]
        (0,0) edge [bend left=13] (#1-n1)
        (#1-n1) edge[bend left=13] (0,0)
        (#1-n1) edge[bend right=13] (#1-n2);
      \draw [fedge]
        (#1-n2) edge[bend right=13] (#1-n1);
    }
  },
}
\begin{figure}
 \centerline{%
 \begin{tikzpicture}[baseline=0]
  \pic at (0,0) {branch=a};
  \pic [rotate=42] at (0,0) {branch=b};
  \pic [rotate=138] at (0,0) {branch=d};
  \pic [rotate=180] at (0,0) {branch=e};
  \node at (0,0.8) {$\ldots$};
  \node at (0,1.2) {$k$ many};
  \coordinate (c-n1) at (-0.75,1.2);
  \coordinate (c-n2) at (0.75,1.2);
  \draw [gedge, not in solution] (a-n1) edge (b-n2);
  \draw [gedge, not in solution] (b-n1) edge (c-n2);
  \draw [gedge, not in solution] (c-n1) edge (d-n2);
  \draw [gedge, not in solution] (d-n1) edge (e-n2);
  \node [filledleaf] at (0,0) {};
  \node at (6,0.6) {%
    \begin{minipage}{45mm}
      $\sum_{v \in L}{\deg_{E'}(v)} = 8k$\\[1ex]
      $|L_0| = 1, \; |L_1| = 2k, \; |L_2| = 0$\\[1ex]
      $3 \, (|L| + |L_0|) + |L_1| = 8k + 6$
    \end{minipage}};
 \end{tikzpicture}}
 \caption{\label{fig:tight}%
  An example showing
   tightness for the approximation ratio as well as for the leaf-degree property.
  For the approximation ratio, consider
   all thick edges' costs to be~$0$,
   all solid thin edges' costs~$1$,
   and all dashed edges' costs~$1+\varepsilon$
   for an arbitrary small~$\varepsilon > 0$.
  The algorithm's solution consists of all solid edges of total cost~$3k$.
  The optimum solution
   is the Hamiltonian cycle consisting of
   all dashed edges,
   all thick edges,
   and two solid thin edges to connect the center node.
  Its total cost is~$
   k + 1 + (k-1) \varepsilon
  $.
  The ratio~$\frac{3k}{k + 1 + (k-1) \varepsilon)}$
   approaches~$3$
   for~$k \to \infty$.
  For the leaf-degree property,
   all edges are in~$E$,
   thick edges in~$F \subsetneq E$,
   solid edges in~$E' \subsetneq E$.
 }
\end{figure}

This section is dedicated to show the following theorem.
The theorem is a reformulation of Lemma~\ref{lemma:degree-property}
 in terms that are totally independent of the setting and notation
 used in the previous section.

\begin{theorem}[Reformulation of Lemma~\ref{lemma:degree-property}]\label{thm:degree-property}
Let $G = (V, E)$ be a 2-edge-connected graph and
 $E'$ a minimal 2-edge-connected spanning subgraph in $G$.
Let $F \subseteq E$
 be an edge set describing a (not necessarily spanning) forest in $G$
 such that
 each edge $e \in E' \setminus F$
 is essential in $E' \cup F$.
Let $L_0 := V \setminus V(F)$,
 $L_1 := \{v \in V(F) \mid \deg_F(v) = 1, \delta_F(v) \subseteq E'\}$,
 $L_2 := \{v \in V(F) \mid \deg_F(v) = 1, \delta_F(v) \cap E' = \varnothing\}$,
and $L := L_0 \cup L_1 \cup L_2$.

Then we have
 $
  \sum_{v \in L}{\deg_{E'}(v)}
   \; \leq \;
   3 \, (|L| + |L_0|) + |L_1|
 $.
\end{theorem}
Figure~\ref{fig:tight} illustrates an example
 where the left-hand side approaches the right-hand side.
Throughout this section,
 we will use the following convention:
We call the nodes in $L$ \emph{leaves};
 they are either $L_1 \cup L_2$, degree-1 nodes in $F$, or
 $L_0$, isolated nodes w.r.t.~$F$.
This is quite natural
 since $E$ and $E'$ do not contain any degree-1 nodes.
For any subforest $F' \subseteq F$,
 let $L(F') := L \cap V(F')$.
We use the term \emph{component}
 for a connected component in $F$
 since $E$ and $E'$ consist of one connected component only.
Hence these two terms only make sense in the context of $F$.

We consider an \emph{ear decomposition} of $E'$, that is, we consider
an ordered partition of $E'$ into disjoint edge sets
$O_0, O_1, \ldots$ where $O_0$ is a simple cycle and where
$O_t$ for $t \geq 1$ is a simple $u$-$v$-path with
$V(O_t) \cap \bigcup_{i=0}^{t-1}{V(O_i)} = \{u, v\}$.
Such an ear decomposition exists since $E'$ is 2-edge-connected.
Note that
 every ear $O_t$ has at least one inner node
 since it would otherwise only consist of a single edge
 which would be nonessential in $E'$.
Let $E_t' := \bigcup_{i=0}^{t-1}{O_i}$
 be the subgraph of $E'$
 that contains of the first $t$ ears
 of the ear sequence.

We interpret the ear decomposition as a sequential procedure.
We say $O_t$ is added to $E_t'$ at \emph{time} $t$.
For $t_2 > t_1 \geq 1$,
 the ear $O_{t_1}$ appears \emph{earlier} than $O_{t_2}$,
 and $O_{t_2}$ appears \emph{later} than $O_{t_1}$.
At any time $t$,
we call a node $v$ \emph{explored}
 if $v \in V(E_t')$,
 otherwise it is \emph{unexplored};
we call a component \emph{discovered}
 if it contains an explored node,
 otherwise it is \emph{undiscovered}.
Observe that the inner nodes $v$ of $O_t$ are not yet explored at time~$t$.
We define $\theta(v):=t$ as the time when $v$ will become explored. Clearly, we have $\theta(v):=0$ for all nodes $v\in O_0$.

The basic idea of our proof is to
 use the ear sequence
 to keep track (over time $t$) of $\deg_{E'_t}(v)$ for $v \in L$
 via a charging argument.
Consider any $t \geq 1$.
An inner (and thus unexplored) node of $O_t$ might be in $L$.
Every such leaf has a degree of $2$ in $E_{t+1}'$.
However, the endpoints of $O_t$ may be explored leaves
 whose degrees increase in $E_{t+1}'$.
We tackle this problem by
 assigning this increase to other leaves
 and making sure that the total assignment to each leaf is bounded.

Let $\Pi$ %
 be the set of all edges in $E'$
 that are incident to a leaf that is simultaneously an endpoint of some ear $O_t$.
We denote the edges in $\Pi$ by
 $\pi_1, \ldots, \pi_{|\Pi|}$
 in increasing time of their ears,
 i.e.,
  for
  $\pi_i \in O_t, \pi_j \in O_{t'}$
  with
  $i < j$
  we have
  $t \leq t'$.
We say $i$ is the \emph{index} of edge $\pi_i \in \Pi$.
To be able to refer to the nodes $V(\pi_i) =: \{a_i, b_i\}$ by index,
 we define
 $a_i$ as the endpoint
 and
 $b_i$ as the inner node
 of the ear containing $\pi_i$.
Note that
 there might be distinct $\pi_i, \pi_j \in \Pi$
  with $\theta(b_i) = \theta(b_j)$
 if both $a_i$ and $a_j$ are leaves (with possibly even $b_i = b_j$).
By $C_i$ we denote the component that contains~$b_i$.

For any index $i$,
 we may have:
 \emph{situation}~\sitE
  if
  $\pi_i$ is an \underline{e}lement of $F$,
 \emph{situation}~\sitU
  if $\pi_i \notin F$
  and $C_i$ is \underline{u}ndiscovered,
 and
 \emph{situation}~\sitD
  if $\pi_i \notin F$
  and $C_i$ is \underline{d}iscovered,
   c.f.~Figure~\ref{fig:situations}.

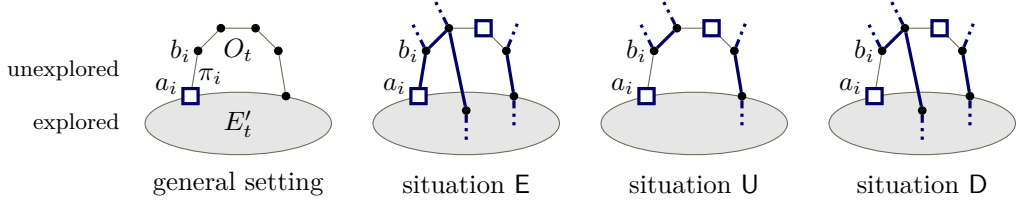
\begin{figure}
 \centerline{%
 \begin{tikzpicture}[%
   pics/setting/.style={
     code = {
       \node (Et) [subgraph] at (0,0.5) {\makebox[1.5cm]{#1}};
       \node (li) [leaf, label={[label distance=1pt, inner sep=0]170:$a_i$}] at (Et.150) {};
       \node (ri) [general] at (Et.30) {};
       \node (ra) [general] at ([shift={(-1mm,6mm)}] ri) {};
       \node (la) [general, label={[label distance=1pt, inner sep=0]180:$b_i$}] at ([shift={(1mm,6mm)}] li) {};
       \node (rb) [general] at ([shift={(-3mm,3mm)}] ra) {};
       \node (lb) [general] at ([shift={(3mm,3mm)}] la) {};
       \draw [gedge] (ri) -- (ra) -- (rb) -- (lb) -- (la) -- (li);
     }
   },
   pics/situation/.style={
     code = {
       \pic {setting={\phantom{$E_t'$}}};
       \node (rb) [leaf] at (rb) {};
       \draw [fedge] (ri) -- (ra);
       \draw [fedge, dash dot dot]
         (ra) -- +(1.5mm,3.75mm)
         (ri) -- +(0,-4.2mm);
       \node at (0,-0.3) {situation~#1};
     }
   }]
  \node [yshift=7mm] at (-0.3,0.5) {\footnotesize unexplored};
  \node at (-0.3,0.5) {\footnotesize \phantom{un}explored};
  \pic at (2,0) {setting={$E_t'$}};
  \node at ([yshift=5.5mm] Et.north) {$O_t$};
  \node at ([shift={(2mm,-3mm)}] la.south west) {$\pi_i$};
  \node at (2,-0.3) {general setting};
  \pic at (5,0) {situation=\sitE};
  \node [general] (ex) at ([yshift=-2.5mm] Et.north) {};
  \draw [fedge] (li) -- (la) -- (lb) -- (ex);
  \draw [fedge, dash dot dot]
    (la) -- +(-2mm,4mm)
    (lb) -- +(-2mm,4mm)
    (ex) -- +(0,-4.2mm);
  \pic at (8,0) {situation=\sitU};
  \draw [fedge] (la) -- (lb);
  \draw [fedge, dash dot dot]
    (la) -- +(-2mm,4mm)
    (lb) -- +(-2mm,4mm);
  \pic at (11,0) {situation=\sitD};
  \node [general] (ex) at ([yshift=-2.5mm] Et.north) {};
  \draw [fedge] (la) -- (lb) -- (ex);
  \draw [fedge, dash dot dot]
    (la) -- +(-2mm,4mm)
    (lb) -- +(-2mm,4mm)
    (ex) -- +(0,-4.2mm);
 \end{tikzpicture}}
 \caption{\label{fig:situations}%
  Illustration of the general setting for an ear $O_t$ and a $\pi_i \in \Pi$
 with $\theta(b_i) = t$,
   and examples of situations~\sitE,~\sitU, and~\sitD.
  Thick edges are in $F$,
  rectangular nodes in $L$.}
\end{figure}

We will assign the degree increments of $a_i$
 to other leaves
 by some charging scheme~$\chi$, which is the sum of several distinct charging schemes.
 The precise definition of these (sub)schemes is subtle and necessarily intertwined with the analysis
 of the schemes' central properties. Thus we will concisely define
 them only within the proofs of Lemmata~\ref{lemma:charge-e}
 and~\ref{lemma:charge-u-and-d} below.
We call a leaf \emph{charged} due to a specific situation
 if that situation applied
 at the time
 when the increment was assigned to the leaf.
Let $\chi_\sitE, \chi_\sitU, \chi_\sitD\colon L \to \mathbb{N}$
 be the overall charges (on a leaf) due to situation \sitE, \sitU, \sitD, respectively.
The leaf-degree property will follow by observing that no leaf is charged too often by these different chargings.

\begin{lemma}\label{lemma:charge-e}
 We can establish a charging scheme $\chi_\sitE$ such that
  we guarantee
  $\chi_\sitE(v) \leq 1$ if $v \in L_1$
  and
  $\chi_\sitE(v) = 0$ if $v \in L_0 \cup L_2$.
\end{lemma}
\begin{proof}
 Consider situation~\sitE occurring for index $i$.
 By $\pi_i \in F$,
  we have $a_i \in L_1$
  (and thus $\chi_\sitE(a_i) = 0$ if $a_i \in L_0 \cup L_2$).
 Assume situation~\sitE occurs for another index $j \neq i$
  such that
  $a_i = a_j$.
 This yields $\pi_i, \pi_j \in F$
  which contradicts that $a_i$ is a leaf.
 Hence the claim follows by setting $\chi_\sitE(a_i) = 1$.
\end{proof}
\begin{lemma}\label{lemma:charge-u-and-d}
 We can establish charging schemes $\chi_\sitU,\chi_\sitD$ such that
  we guarantee
  $\chi_\sitU(v) + \chi_\sitD(v) \leq 2$ if $v \in L_0$
  and
  $\chi_\sitU(v) + \chi_\sitD(v) \leq 1$ if $v \in L_1 \cup L_2$.
\end{lemma}
The proof is rather technical and will be proven in the following subsection.
It mainly exploits
 the finding of contradictions to the fact
 that each edge $e \in E' \setminus F$ is essential in $E' \cup F$.
Two mappings can be established:
 first an injective mapping
  (based on induction)
 from edges $\pi_i$ in situation~\sitD to leaves,
 and second an `almost injective' (relaxing the mappings to $L_0$ nodes slightly) mapping
 from edges $\pi_i$ in situation~\sitU
 to remaining leaves.
For the latter,
 we establish an algorithm that hops through components.
We show that this algorithm identifies suitable distinct leaves.
The charging schemes $\chi_\sitU, \chi_\sitD$
 with the desired properties
 follow from these mappings.
\begin{proof}[Proof of Theorem~\ref{thm:degree-property}]
 Let $v \in L$ be any leaf.
 The charging of $v$ during the whole process
  is $\chi(v) := 2 + \chi_\sitE(v) + \chi_\sitU(v) + \chi_\sitD(v)$
  where the $2$ comes from an implicit charging of the degree of $v$
  when $v$ is discovered.
 By Lemmata~\ref{lemma:charge-e} and~\ref{lemma:charge-u-and-d},
  we obtain
  $\chi(v) \leq 4$ for $v \in L_0$,
  $\chi(v) \leq 4$ for $v \in L_1$,
  and
  $\chi(v) \leq 3$ for $v \in L_2$.
 This yields
 $\sum_{v \in L}{\deg_{E'}(v)}
  \leq
   4 |L_0| + 4 |L_1| + 3 |L_2|
  \leq
   3 (|L| + |L_0|) + |L_1|$.
\end{proof}

\subsection{Proof of Lemma~\ref{lemma:charge-u-and-d}}

We first introduce some notation
 in order to show Lemma~\ref{lemma:charge-u-and-d}.
For any subforest $F' \subseteq F$
 and $S \in \{\sitU, \sitD\}$,
 let $\PiX{S}(F') := \{\pi_j \in \Pi \mid{}$situation~$S$ applies for $j$ with $b_j \in V(F')\}$.
Let
 $\aX{S}(F') := \{a_j \mid \pi_j \in \PiX{S}(F')\}$.
 For any subgraph $H$ in $G$ and two (not necessarily distinct) nodes
 $x_1, x_2 \in V(H)$ we define $\pathset{H}{x_1}{x_2}$ to be the set
 of all paths in $H$ between $x_1$ and~$x_2$.
Consider nodes $w_0,w_1,\dots,w_k \in V$
 for some $k \in \mathbb{N}$
 and a collection $\mathcal{P}_1,\dots,\mathcal{P}_k$
 of $w_{j-1}$-$w_j$-paths, that is,
 $\mathcal{P}_j \subseteq \pathset{G}{w_{j-1}}{w_j}$ for
 $j=1,\dots,k$.
Then let
 $\buildpath{\mathcal{P}_1, \ldots, \mathcal{P}_k}$ denote the set
 of all $w_0$-$w_k$-paths that are the concatenation of $k$
 (necessarily) pairwise disjoint paths $P_1,\dots,P_k$ with
 $P_j\in\mathcal{P}_j$ for $j=1,\dots,k$.
For notational simplicity,
 we may also use single paths and single edges as sets
 $\mathcal{P}_j$.  Note that
 $\buildpath{\pathset{E_t'}{x_1}{a_i},
   \pathset{E_t'}{a_i}{x_2}} \neq \varnothing$ for nodes
 $x_1, x_2 \in V(E_t')$, which follows from the well-known fact that
 any 2-edge-connected graph contains two disjoint $u$-$v$-paths for
 all nodes $u, v$.

The below proofs of our auxiliary lemmas will use the following reasoning.
We will determine an edge $\pi_j \in \Pi \setminus F$
 and a set of paths $\mathcal{P}$ such that there is a cycle
 $Q \in \mathcal{P}$ with $V(\pi_j) \subseteq V(Q)$, and
 $\pi_j \notin Q$.
Since $\pi_j \notin F$, we say that
 $\cyclewitness{j}{\mathcal{P}}$ is a \emph{cycle witness} that
 contradicts our assumption that every edge in $E' \setminus F$ is
 essential in $E' \cup F$.

\begin{lemma}\label{lemma:situation-a-unexplored-leaf}
 Let $\pi_i \in \PiD(F)$
  and $t := \theta(b_i)$.
 There are no two disjoint paths in $C_i$ between $b_i$ and explored nodes.
 Moreover,
  there is at least one unexplored leaf in $C_i$.
\end{lemma}
\begin{proof}
 Assume
  there are two disjoint paths
  $P_1, P_2$
  between $b_i$ and nodes $w_1, w_2 \in V(E_t') \cap V(C_i)$, respectively.
 We can w.l.o.g.\ assume that $P_1, P_2$ do not contain explored nodes other than $w_1, w_2$, respectively.
 Then $\badcycle{i}{P_1, \pathset{E_t'}{w_1}{a_i}, \pathset{E_t'}{a_i}{w_2}, P_2}$ is a cycle witness.

 The second claim follows directly
  as $b_i$ is either an unexplored leaf itself
  or there is a path to another leaf
  that must be unexplored by the first claim.
\end{proof}
Consider component $C_i$ for $\pi_i \in \PiD(F)$.
Based on the above lemma, we define
$\bar C_i$ as the unique path in $\pathset{(C_i \setminus E_t')}{b_i}{y}$
 where $y \in V(E_t') \cap V(C_i)$ is an explored node.
Furthermore, let $C_i^*$ be the component in $C_i \setminus \bar C_i$
 that contains $b_i$.

\begin{lemma}\label{lemma:path-witness}
 Let $S \in \{\sitU, \sitD\}$,
  $\pi_i \in \PiX{S}(F)$,
  $t := \theta(b_i)$,
  $H^\sitU := (E' \cup F) \setminus E_{t+1}'$,
  and
  $H^\sitD := (E' \cup F) \setminus (E_t' \cup \{\pi_i\} \cup \bar C_i)$.
 We have $\pathset{H^S}{b_i}{x} = \varnothing$
  for any $x \in V(E_t')$.
\end{lemma}
\begin{proof}
 Assume there is a path $P \in \pathset{H^S}{b_i}{x}$.
 Let $Q^\sitU := O_t \setminus \{\pi_i\}$
  and $Q^\sitD := \bar C_i$.
 We have a cycle witness
  $\badcycle{i}{P, \pathset{E_t'}{x}{a_i}, \pathset{E_t'}{a_i}{y}, Q^S}$
  with $y \in V(Q^S) \cap V(E_t')$.
\end{proof}
This allows us to define a \emph{path witness}
 $\pathwitness{j}{\mathcal{P}}$
 for a situation $S \in \{\sitU, \sitD\}$
 as a shorthand for a cycle witness
 on edge $\pi_j$ with $P \in H^S$
 in the proof of Lemma~\ref{lemma:path-witness}.

\begin{lemma}\label{lemma:situation-a-component-mapable}
 Let $\pi_i \in \PiD(F)$.
 We have $|\PiD(C_i^*)| \leq |\aD(C_i^*) \cup L(C_i^*)| - 1$.
\end{lemma}
\begin{proof}
 Note that $a_i \notin L(C_i^*)$.
 Let $r := |L(C_i^*)|$,
  and
  $x_1, x_2, \ldots, x_r$
  the members of $L(C_i^*)$
  such that
   for each $j \in \{2, \ldots, r\}$
   we have $\theta(x_j) \geq \theta(x_{j-1})$.
 Let $T_1$ be the path in $\pathset{C_i^*}{b_i}{x_1}$.
 Given $T_{j-1}, j \in \{2, \ldots, r\}$,
  we obtain $T_j$
  by adding a path $P_j \in \pathset{(C_i^* \setminus T_j)}{h_j}{x_j}$
  where $h_j \in V(T_j)$. %
 Note that
  $V(P_j) \cap V(T_{j-1}) = \{h_j\}$
  and
  $L(T_j) = L(T_{j-1}) \cup \{x_j\}$.
 By definition, we have $T_r = C_i^*$.
 For brevity,
  let
  $P_j' := P_j - h_j$
  and
  $a_j := \aD(T_j) \cup L(T_j)$
  for any~$j$.

\begin{claim}\label{claim:mapping:subpath-independent}
 Let $P \in \pathset{C_i^*}{b_i}{x}$ for any $x \in L(C_i^*)$.
 Let $Q \subseteq P$ be any subpath of~$P$.
 Then $\PiD(Q)$ is independent,
 i.e.,
 no two of its edges have a common node.
\end{claim}
\begin{claimproof}
 There are three cases.
 (1)~Assume there are distinct $\pi_k, \pi_\ell \in \PiD(Q)$
  with $b_k = b_\ell$.
 Let $t := \theta(b_\ell)$.
 For $w_1, w_2$ being the endpoints of $O_t$,
  we have a cycle witness $\badcycle{k}{O_t, \pathset{E_t'}{w_1}{w_2}}$.
 (2)~Assume there are distinct $\pi_k, \pi_\ell \in \PiD(Q)$
  with $b_k = a_\ell$.
 Since $\pi_k \in \Pi$, we have $b_k \in L$,
  i.e.,
  $Q$ ends at $b_k$
  and $b_k = x$.
 Hence,
  $\pathset{Q}{b_\ell}{b_k}$
  contradicts Lemma~\ref{lemma:situation-a-unexplored-leaf}
  since $\theta(b_k) = \theta(a_\ell) = t$.
 (3)~Assume there are distinct
  $\pi_k, \pi_\ell \in \PiD(Q)$
  with
  $a_k = a_\ell$,
  w.l.o.g.\ $k < \ell$.
 We have a path~witness $\badpath{k}{\pathset{Q}{b_k}{b_\ell}, \pi_\ell}$.%
\end{claimproof}
\begin{claim}\label{claim:mapping:subpath-leaf-is-not-u}
 Let $P \in \pathset{C_i^*}{b_i}{x}$ for any $x \in L(C_i^*)$.
 Let $Q \subseteq P$ be any subpath of~$P$ with $x \in V(Q)$.
 Then $x \notin \aD(Q)$.
\end{claim}
\begin{claimproof}
 Assume not.
 We have $x \in \aD(O_{\theta(b_k)})$ for some index $k > i$.
 Hence $x \in V(C_k^*)$ is an explored node at time $\theta(b_k)$
  which contradicts Lemma~\ref{lemma:situation-a-unexplored-leaf}.
\end{claimproof}
\begin{claim}\label{claim:mapping:one-common}
 For each $j \in \{2, \ldots, r\}$, we have
  $|\aD(P_j') \cap a_{j-1}| \leq 1$.
\end{claim}
\begin{claimproof}
 Assume not.
 Let $Q$ be the path in $T_{j+1}$ between $b_i$ and $x_j$.
 $\PiD(Q)$ is independent by Claim~\ref{claim:mapping:subpath-independent}.
 Hence
  there are
  $\pi_{\ell_1}, \pi_{\ell_2} \in \PiD(Q)$
  and
  $v_1, v_2 \in V(T_{j-1} \setminus Q), v_1 \neq v_2,$
  such that
  one of the following holds:
  (1)~$\mathcal{A}_1$ and~$\mathcal{A}_2$,
  (2)~$\mathcal{A}_1$ and~$\mathcal{B}_2$,
  (3)~$\mathcal{B}_1$ and~$\mathcal{A}_2$,
  (4)~$\mathcal{B}_1$ and~$\mathcal{B}_2$,
 where
 $\mathcal{A}_d$, $d=1,2$, is the case that
  there is a $\pi_{k_d} \in \PiD(T_{j-1} \setminus Q)$
  with $b_{k_d} = v_d$
  and $a_{\ell_d} = a_{k_d}$,
 and $\mathcal{B}_d$ is the case that
  we have $v_d \in L(T_{j-1})$
  with $a_{\ell_d} = v_d$.
 W.l.o.g.\ $\ell_1 \leq \ell_2$.

 For case~(1), we have
  a path witness
  $\badpath{\ell_1}{\pathset{P_j'}{b_{\ell_1}}{b_{\ell_2}}, \pi_{\ell_2}, \pi_{k_2}, \pathset{T_{j-1}}{v_2}{v_1}, \pi_{k_1}}$.

 For case~(2), we have
  a path witness
  $\badpath{\ell_1}{\pathset{P_j'}{b_{\ell_1}}{b_{\ell_2}}, \pi_{\ell_2}, \pathset{T_{j-1}}{v_2}{v_1}, \pi_{k_1}}$.

 For case~(3), we have
  a path witness
  $\badpath{\ell_1}{\pathset{P_j'}{b_{\ell_1}}{b_{\ell_2}}, \pi_{\ell_2}, \pi_{k_2}, \pathset{T_{j-1}}{v_2}{v_1}}$.

 For case~(4), we have
  a path witness
  $\badpath{\ell_1}{\pathset{P_j'}{b_{\ell_1}}{b_{\ell_2}}, \pi_{\ell_2}, \pathset{T_{j-1}}{v_2}{v_1}}$.
\end{claimproof}
\begin{claim}\label{claim:mapping:tree-leaf}
 For each $j \in \{2, \ldots, r\}$, we have
  $x_j \notin a_{j-1}$.
\end{claim}
\begin{claimproof}
 We have $x_j \notin L(T_{j-1})$ by definition of $T_{j-1}$.
 It remains to show $x_j \in \aD(T_{j-1})$.
 Assume not.
 There is a $\pi_\ell \in \PiD(T_{j-1})$
  with $a_\ell = x_j$.
 Choose $x_k \in L(T_{j-1})$
  such that $b_\ell$ lies on the path between $b_i$ and $x_k$.
 By definition of $x_j$,
  we have $\theta(x_j) > \theta(x_k)$.
 Lemma~\ref{lemma:situation-a-unexplored-leaf}
  at time $\theta(b_\ell)$
  gives $\theta(x_k) > \theta(b_\ell)$.
 By $a_\ell = x_j$,
  we have $\theta(b_\ell) > \theta(x_j)$,
  a contradiction.
\end{claimproof}

 We show
  $|\PiD(T_j)| \leq |a_j| - 1$
  inductively
  for all $j \in \{1, \ldots, r\}$.
 First consider $j = 1$.
 Since $T_1$ is a path,
  $\PiD(T_1)$ is independent
  by Claim~\ref{claim:mapping:subpath-independent};
  hence
  $|\PiD(T_1)| = |\aD(T_1)|$.
 The claim follows by observing that
  $|\aD(T_1)|
   = |\aD(T_1) \cup \{x_1\}| - 1
   = |a_1| - 1$
   since $x_1 \notin \aD(T_1)$
   by Claim~\ref{claim:mapping:subpath-leaf-is-not-u}.
 We now assume that the claim holds for~$j - 1$ with~$j \in \{2, \ldots, r\}$,
  and show that it holds for~$j$.
 We get
 \begin{align*}
  |\PiD(T_j)|
  &
   = |\PiD(P_j) \cup \PiD(T_{j-1})|
   = |\PiD(P_j')| + |\PiD(T_{j-1})|
  \\ &
   \leq |\PiD(P_j')| + |a_{j-1}| - 1
  & & \text{by induction}
  \\ &
   = |\aD(P_j')| + |a_{j-1}| - 1
  & & \text{by Claim~\ref{claim:mapping:subpath-independent}}
  \\ &
   \leq |\aD(P_j')| + |a_{j-1}| - |\aD(P_j') \cap a_{j-1}|
  & & \text{by Claim~\ref{claim:mapping:one-common}}
  \\ &
   = |\aD(P_j') \cup a_{j-1}|
  \\ &
   = |\aD(P_j') \cup a_{j-1} \cup \{x_j\}| - 1
  & & \text{by Claims~\ref{claim:mapping:subpath-leaf-is-not-u} and~\ref{claim:mapping:tree-leaf}}
  \\ &
   = |\aD(P_j') \cup \aD(T_{j-1}) \cup L(T_j)| - 1
  \\ &
   = |\aD(T_j) \cup L(T_j)| - 1
   = |a_j|- 1
   \text{.}
   & & \qedhere
 \end{align*}
\end{proof}

For each $\pi_i \in \PiD(F)$,
 let $\sidx{i} := \min\{j \mid \pi_i \in \PiD(C_j^*)\}$
 be the index of the earliest situation~\sitD on component $C_i$.
Let $\startsD := \{\sidx{i} \mid \pi_i \in \PiD(F)\}$.
Using Lemma~\ref{lemma:situation-a-component-mapable},
 we construct an injection $\mapD_i\colon \PiD(C_i^*) \to L \setminus \{a_i\}$
 for every $i \in \startsD$.
First observe that
 $|\PiD(C_i^*)| \leq |(\aD(C_i^*) \setminus \{a_i\}) \cup L(C_i^*)|$
 by
 $a_i \in \aD(C_i^*)$
 and
 $a_i \notin L(C_i^*)$ (see Lemma~\ref{lemma:situation-a-unexplored-leaf}).
There might be distinct $\pi_j, \pi_k \in \PiD(C_i^*)$
 with $a_j = a_k$.
It is possible to construct
 $\mapD_i$ %
 as injection
 such that
 for each $w \in \aD(C_i^*) \setminus \{a_i\}$
 there is one $k$
 with $w = a_k$
 and
 $\mapD_i(\pi_k) = a_k$.
Since components are a partition of $F$,
 we can define a mapping $\mapD\colon \PiD(F) \to L$
 by $\mapD := \bigcup_{j \in \startsD}{\mapD_j}$.

\begin{lemma}\label{lemma:situation-a-mapping}
 The mapping $\mapD$ is an injection.
\end{lemma}
\begin{proof}
 Assume there are $\pi_k, \pi_\ell \in \PiD(F)$
  with $\sidx{k} < \sidx{\ell}$
  and $w := \mapD_\sidx{k}(\pi_k) = \mapD_\sidx{\ell}(\pi_\ell)$.
 We have $C_\sidx{k}^* \neq C_\sidx{\ell}^*$
  since otherwise $\mapD_\sidx{k}(\pi_k) = \mapD_\sidx{k}(\pi_\ell)$
  contradicts the injectivity of $\mapD_\sidx{k}$.
 The following three cases remain:
  (1)~$w = a_\ell \in L(C_\sidx{k}^*)$,
  (2)~$w = a_k \in L(C_\sidx{\ell}^*)$,
  and
  (3)~$w = a_k = a_\ell$.
 Consider case~(1).
 By $\ell \neq \sidx{\ell}$
  and $k < \ell$ (since $w = a_\ell$),
 we have a path witness
  $\badpath{\sidx{\ell}}{%
   \pathset{C_\sidx{\ell}^*}{b_\sidx{\ell}}{b_\ell},
   \pi_\ell,
   \pathset{C_\sidx{k}^*}{w}{b_\sidx{k}},
   \pi_\sidx{k}}$.
 Consider case~(2).
 By $k \neq \sidx{k}$
  and $\ell < k$ (since $w = a_k$),
 we have a path witness
  $\badpath{\sidx{\ell}}{%
   \pathset{C_\sidx{\ell}^*}{b_\sidx{\ell}}{w},
   \pi_k,
   \pathset{C_\sidx{k}^*}{b_k}{b_\sidx{k}},
   \pi_\sidx{k}}$.
 For case~(3),
 we have a path witness
  $\badpath{\sidx{\ell}}{%
   \pathset{C_\sidx{\ell}^*}{b_\sidx{\ell}}{b_\ell},
   \pi_\ell,
   \pi_k,
   \pathset{C_\sidx{k}^*}{b_k}{b_\sidx{k}},
   \pi_\sidx{k}}$.
\end{proof}

For any $F' \subseteq F$,
 let $L'(F') := L(F') \setminus \mapD(\PiD(F))$
 be the leaves not used by $\mapD$.

\begin{lemma}\label{lemma:situation-e-mapping}
 There is a mapping $\mapU\colon \PiU(F) \to L'(F)$
  such that
  for each $v \in L$,
  we have
  $|\mapU^{-1}(v)| \leq 2$ if $v \in L_0$
  and
  $|\mapU^{-1}(v)| \leq 1$ otherwise.
\end{lemma}
\begin{proof}
 We give an algorithm that establishes our mapping $\eta$.
 Let $C \subseteq F$ be a subtree
  and $w \in V(C)$.
 Consider the following recursive algorithm
  which,
  invoked on $(C, w)$,
  tries to construct a path $P$
  between $w$
  and a leaf $x \in L'(F)$.
 $P$ is initially empty and will be extended in each recursion step.
 Trivially, if there is an $x \in L'(C)$,
  the algorithm
  adds the unique path in $\pathset{C}{w}{x}$ to $P$
  and terminates.
 Otherwise,
  we have $L'(C) = \varnothing$.
 There are two cases:
 \begin{enumerate}
 \item
  There is a $\pi_k \in \PiD(C)$.
  Let $C'$ be the component containing $a_\sidx{k}$.
  We add to $P$ the unique path in $\buildpath{\pathset{C}{w}{b_\sidx{k}}, \pi_\sidx{k}}$
   and recurse on $(C', a_\sidx{k})$.
 \item
  We have $\PiD(C) = \varnothing$
  but then
  there is a component $C' \neq C$
   with $\pi_k \in \PiD(C')$ and~$\mapD(\pi_k) = a_k \in L(C)$.
  We add to $P$ the unique path in $\buildpath{\pathset{C}{w}{a_k}, \pi_k}$
   and recurse on~$(C', b_k)$.
 \end{enumerate}
 We define $C_i^*$ for a given $\pi_i \in \PiU(F)$
  to be the component in $C_i \setminus O_{\theta(b_i)}$
  that contains $b_i$.
 However, there is a tricky exception:
  if $\pi_{i+1} \in \PiU(F)$
   and $O_{\theta(b_i)} = \{\pi_i, \pi_{i+1}\}$,
   i.e.,
   if we have $b_i = b_{i+1}$.
 Then, if $b_i \in L_0$,
  we say that $C_i^*$ and $C_{i+1}^*$ consist only of $b_i$
  ($\eta$ may map to it twice anyhow).
 Otherwise,
  we have at least two leaves in $C_i = C_{i+1}$.
 By removing an edge $e$ with $V(e) = \{b_i, z\}$ from $C_i$,
  we obtain the two components
  $C_i^*$ and $C_{i+1}^*$
  such that
   $b_i \notin V(C_i^*)$
   and
   $b_i \in V(C_{i+1}^*)$.
 For technical simplicity,
  we set $b_i' := b_i$ if $b_i \in C_i^*$
  and $b_i' := z$ otherwise.
 However, $\eta$ will never map to $z$ in the following.

 We now construct $\mapU$
  by invoking the algorithm
  on $(C_i^*, b_i')$
  for each $\pi_i \in \PiU(F)$
  in chronological order from the latest to the earliest component;
  we set $\mapU(\pi_i) := x$ where $x$ is the found leaf.
 By construction,
  $b_i$ is the earliest node in $P$,
  since otherwise we would have a path witness $\pathwitness{b_i}{P}$.

 Assume by contradiction that the algorithm does not terminate.
 Consider the recursion step
  where edges are added to $P$ that are already included in $P$.
 $P$ contains a cycle $Q$.
 Note that in case~(1),
  we have
  $\theta(a_\sidx{k}) < \theta(w)$,
  i.e.,
  we go back in time only,
  and thus
  $Q$ also involves a case~(2) step.
 On the other hand,
  after a recursion step handling case~(2),
  we either terminate
  or recurse into case~(1).
 Hence there is a component $C$
  such that
  by case~(2) there is a $\pi_j \in \PiD(C)$
   with $b_j \in V(C)$
   and $j \notin \startsD$
  and
  by case~(1) we have $b_\sidx{j} \in V(C)$.
 Now
  $Q' := Q \setminus \{\pi_\sidx{j}\}$ is a path
  between $b_\sidx{j}$ and $a_\sidx{j}$;
  $\pathwitness{\sidx{j}}{Q'}$
  is a path witness.

 Now that we can ensure that the algorithm terminates,
 consider an arbitrary $\pi_i \in \PiU(F)$.
 For an invocation of the algorithm on $(C_i^*, b_i')$,
  let $x_i \in L'(F)$ be the resulting leaf
  and $P_i$ the resulting $b_i$-$x_i$-path.
 Assume that
  there is a $\pi_j \in \PiU(F)$
  with
  $x_i \in \mapU(\pi_j)$.
 Note that $j > i$
  since we invoke the recursive algorithm from last to first index.
 Let $y$ be the first node
  that $P_i$ and $P_j$ have in common,
  and let
  $P'$ be the unique path in $\buildpath{\pathset{P_j}{b_j}{y}, \pathset{P_i}{y}{b_i}}$.
 If $\theta(b_j) > \theta(b_i)$,
  we have a path witness
  $\badpath{j}{P', \pi_i}$.
 Now consider the case $\theta(b_j) = \theta(b_i)$.
 Let $C$ be the component containing $y$.
 We distinguish the following subcases:
 \begin{itemize}
  \item
   If $L'(C) = L(C)$,
    we have $x_i \in L(C)$.
   If $x_i \in L_0$,
    we set $\eta(b_i) := x_i$.
   Otherwise there is an $x_j \in L(C) \setminus \{x_i\}$.
   Since $\mapU(b_j)$ is already set to $x_i$,
    we set $\mapU(b_i) := x_j$.
  \item
   If $P'$ enters $C$ using case~(1),
    there is a $\pi_k \in \PiD(P')$
    with $k \in \startsD$
    and $a_k \in V(C)$.
   Then
    $\pathwitness{k}{\pathset{P'}{b_k}{b_i}}$
    or
    $\pathwitness{k}{\pathset{P'}{b_k}{b_j}}$
    is a path witness
    since $\theta(b_k) > \theta(b_i) = \theta(b_j)$.
  \item
   If $P'$ enters $C$ using case~(2),
    there is a $\pi_k \in \PiD(P')$
    with $k \notin \startsD$
    and $b_k \in V(C)$.
   Hence there is a $\pi_\sidx{k} \in \PiD(C)$
    and
    $\badpath{\sidx{k}}{\pathset{P'}{b_\sidx{k}}{b_i}}$
    or
    $\badpath{\sidx{k}}{\pathset{P'}{b_\sidx{k}}{b_j}}$
    is a path witness
    since $\theta(b_\sidx{k}) > \theta(b_i) = \theta(b_j)$.
    \qedhere
 \end{itemize}
\end{proof}

It is now easy to show
 Lemma~\ref{lemma:charge-u-and-d}
 using
 Lemmata~\ref{lemma:situation-a-mapping}
 and~\ref{lemma:situation-e-mapping}.
We first charge all situation~\sitD nodes,
 that is,
 we set $\chi_\sitD(v) = 1$ for all $v \in \mapD(\PiD(F))$.
Now we charge all situation~\sitU nodes using $\mapU$,
 that is,
 we have $\chi_\sitU(v) = 1$ for all $v \in L'(F) \setminus L_0$
 and $\chi_\sitU(v) \leq 2$ for all $v \in L'(F) \cap L_0$.

\section{Conclusion}

We presented a simple 3-approximation algorithm for 2ECSS with general edge costs.
While there have been primal-dual approximations before (but none achieving a better ratio based on the primal-dual method),
 they require two grow phases (followed by a cleanup phase)
 to first compute a tree
 and then augment this tree to a 2-edge-connected solution.
Our approach does not require this separation,
 by the (to our best knowledge) new idea of growing the solution only at leaves.

While our primal-dual analysis is non-trivial,
 the resulting algorithm is very straight-forward to implement
 with $\bigO(\min\{nm,m + n^2 \log n\})$ time,
 requiring only very basic graph operations and the simplest data structures.
An implementation with time $\bigO(nm)$ is remarkably simple.
This is in contrast to the other known primal-dual algorithms.
Of those, only the algorithm in~\cite{GGW98} achieves a faster running time of
 $\bigO(n^2+n\sqrt{m\log\log n})$, but at the cost of requiring
 intricate data structures and subalgorithms detailed in separate papers~\cite{GGST86,WGMV95}.
For sparse graphs, $m\in\bigO(n)$, our running time
 is in fact equivalent. For dense simple graphs, $m\in\bigO(n^2)$, we have $\bigO(n^2\log n)$
 instead of their $\bigO(n^2\sqrt{\log\log n})$.

Note that on instances with uniform costs,
 the ratio naturally drops to the trivial approximation ratio~2.
We may also note that former tight examples (for example, the tight instances for the 3-approximation given in \cite{FJ81,KR93})
 are now approximated with factor 2.
Moreover,
 by a simple extension,
 our algorithm can also compute lower bounds
 (which could be useful for branch-and-bound algorithms and instance preprocessing),
 without changing its runtime complexity.

It would be interesting to see if (and how) it is possible
 to improve our approach to achieve an even better running time or approximation ratio,
 and/or
 to transfer it to 2-node-connectivity
 or generalized edge-connectivity (e.g., $\{0,1,2\}$-survivable network design) problems.

\bibliography{refs}
\end{document}